\documentclass[submission,copyright,creativecommons]{eptcs}
\providecommand{\event}{1st International Workshop on Strategic Reasoning}

\usepackage{graphicx}
\usepackage{amssymb}
\usepackage{amsmath}
\usepackage{wrapfig}
\usepackage{multicol,caption}
\usepackage{lipsum}

\newtheorem{lemma}{Lemma}
\newtheorem{proposition}{Proposition}
\newtheorem{definition}{Definition}

\newenvironment{proof}{\noindent{\sf Proof.}}{\hfill $\boxtimes\hspace{2mm}$\linebreak}
\newcommand{\qed}{\hfill $\boxtimes\hspace{2mm}$ \linebreak}
\begin{document}

\title{Information Flow in Strategic Games}
\title{Functional Dependence in Strategic Games \\ (extended abstract)}
\author{ Kristine Harjes and Pavel Naumov\\ \\ \small Department of Mathematics and Computer Science\\ \small McDaniel College, 
Westminster, Maryland, USA \\ \small
              {\sf \{keh013,pnaumov\}@mcdaniel.edu}     }

\def\titlerunning{Functional Dependence in Strategic Games}
\def\authorrunning{Kristine Harjes and Pavel Naumov}

\maketitle

\begin{abstract}
The paper studies properties of functional dependencies between strategies of players in Nash equilibria of multi-player strategic games. The main focus is on the properties of functional dependencies in the context of a fixed dependency graph for pay-off functions. A logical system describing properties of functional dependence for any given graph is proposed and is  proven to be complete.
\end{abstract}




\section{Introduction}

\vspace{2mm}
\noindent {\bf Functional Dependence.}
In this paper we study dependency between players' strategies in Nash equilibria. For example, the coordination game described by Table~1 has two Nash equilibria: $(a_1,b_1)$ and $(a_2,b_2)$. Knowing the strategy of player $a$ in  a Nash equilibrium of this game, one can predict the strategy of player $b$. We say that player $a$ functionally determines player $b$ and denote this by $a\rhd b$. 

\begin{wraptable}{l}{0.3\textwidth}
\begin{center}
\vspace{0mm}
\begin{tabular}{l|c|c|c}
		& $b_1$ 	& $b_2$\\ \hline
$a_1$	& 1,1	& 0,0\\ 
$a_2$	& 0,0	& 1,1\\ 
\end{tabular}
\caption{Coordination Game}
\end{center}
\label{coordination game}
\vspace{2mm}
\end{wraptable}

Note that in the case of the coordination game, we also have $b\rhd a$. However, for the game described by Table~\ref{asymmetric game} statement $a\rhd b$ is true, but $b\rhd a$ is false.

\begin{wraptable}{r}{0.3\textwidth}
\vspace{-7mm}
\begin{center}
\begin{tabular}{l|c|c|c}
		& $b_1$ 	& $b_2$\\ \hline
$a_1$	& 1,1	& 0,0\\ 
$a_2$	& 0,0	& 1,1\\ 
$a_3$	& 1,1	& 0,0\\ 
\end{tabular}
\end{center}
\caption{Strategic Game}
\label{asymmetric game}
\vspace{-2mm}
\end{wraptable}

The main focus of this paper is functional dependence in multiplayer games. For example, consider a ``parity" game with three players $a$, $b$, $c$. Each of the players picks 0 or 1, and all players are rewarded if the sum of all three numbers is even. This game has four different Nash equilibria: $(0,0,0)$, $(0,1,1)$, $(1,0,1)$, and $(1,1,0)$. It is easy to see that knowledge of any two players' strategies in a Nash equilibrium reveals the third. Thus, using our notation, for example $a,b\rhd c$. At the same time, $\neg(a\rhd c)$. 

As another example, consider a game between three players in which each player  picks 0 or 1 and all players are rewarded if they have chosen the same strategy. This game has only two Nash equilibria: $(0,0,0)$ and $(1,1,1)$. Thus, knowledge of the strategy of player $a$ in a Nash equilibrium reveals the strategies of the two other players. We write this as $a\rhd b,c$.

Functional dependence as a relation has been studied previously, especially in the context of database theory. Armstrong~\cite{a74} presented the following sound and complete axiomatization of this relation:
\begin{enumerate}
\item 
{\em Reflexivity}: $A\rhd B$, if $B\subseteq A$,
\item 
{\em Augmentation}: $A\rhd B \rightarrow A,C\rhd B,C$,
\item 
{\em Transitivity}: $A\rhd B \rightarrow (B\rhd C \rightarrow A\rhd C)$,
\end{enumerate}
where here and everywhere below $A,B$ denotes the union of sets $A$ and $B$. The above axioms are known in database literature as Armstrong's axioms~\cite{guw09}. 
Beeri, Fagin, and Howard~\cite{bfh77} suggested a variation of Armstrong's axioms that describe properties of multi-valued dependence. 

\vspace{2mm}
\noindent {\bf Dependency Graphs.}
As a side result, we will show that the logical system formed by the Armstrong axioms is sound and complete with respect to the strategic game semantics. Our main result, however, is a sound and complete axiomatic system for the relation $\rhd$ in games with a given dependency graph.

Dependency graphs~\cite{kls01uai, lks01nips, egg07ec,egg06eccc} put restrictions on the pay-off functions that can be used in the game. For example, dependency graph $\Gamma_1$ depicted in Figure~\ref{intro_alpha}, specifies that the pay-off function of player $a$ only can depend on the strategy of player $b$ in addition to the strategy of player $a$ himself.  The pay-off function for player $b$ can only depend on the strategies of players $a$ and $c$  in addition to the strategy of player $b$ himself, etc. 

\begin{wrapfigure}{l}{0.35\textwidth}
\begin{center}
\vspace{-5mm}
\scalebox{.5}{\includegraphics{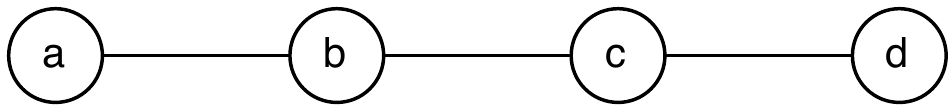}}
\vspace{0mm}
\footnotesize\caption{Dependency Graph $\Gamma_1$}\label{intro_alpha}
\vspace{-5mm}
\end{center}
\vspace{0cm}
\end{wrapfigure}

An example of a game over graph $\Gamma_1$ is a game between players $a$, $b$, $c$, and $d$ in which these players choose real numbers as their strategies. The pay-off function of players $a$ and $d$ is the constant 0. Player $b$ is rewarded if his value is equal to the mean of the values of players $a$ and $c$. Player $c$ is rewarded if his value is equal to the mean of the values of players $b$ and $d$. Thus, Nash equilibria of this game are all quadruples $(a,b,c,d)$ such that $2b=a+c$ and $2c=b+d$. Hence, in this game $a,b\rhd c,d$ and $a,c\rhd b,d$, but $\neg(a\rhd b)$.

Note that although the statement $a,b\rhd c,d$ is true for the game described above, it is not true for many other games with the same dependency graph $\Gamma_1$. In this paper we study properties of functional dependence that are common to all games with the same dependency graph. An example of such statement for the graph $\Gamma_1$, as we will show in Proposition~\ref{XYZ}, is $a\rhd d \rightarrow b,c\rhd d$. 

\begin{wrapfigure}{r}{0.35\textwidth}
\begin{center}
\vspace{0mm}
\scalebox{.5}{\includegraphics{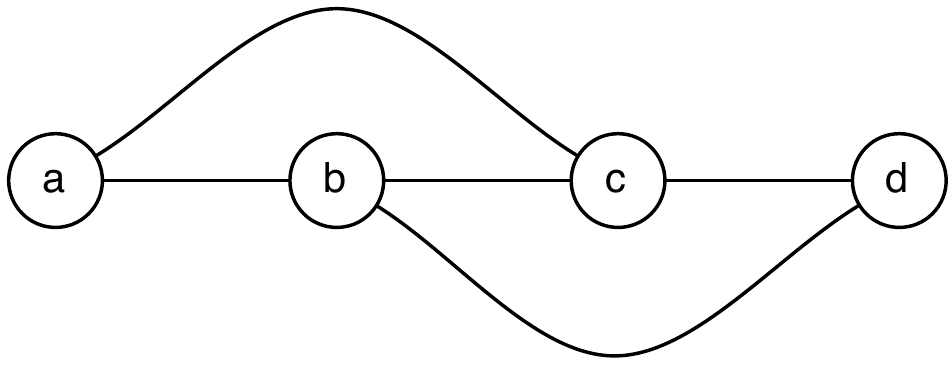}}
\vspace{0mm}
\footnotesize\caption{Dependency Graph $\Gamma_2$}\label{intro_beta}
\vspace{-5mm}
\end{center}
\vspace{0cm}
\end{wrapfigure}

Informally, this property is true for any game over graph $\Gamma_1$ because any dependencies between players $a$ and $d$ must be established through players $b$ and $c$. This intuitive approach, however, does not always lead to the right conclusion. For example, in graph $\Gamma_2$ depicted in Figure~\ref{intro_beta}, players $b$ and $c$ also separate players $a$ and $d$. Thus, according to the same intuition, the statement $a\rhd d \rightarrow b,c\rhd d$ must also be true for any game over graph $\Gamma_2$. This, however, is not true. Consider, for example, a game in which all four players have three strategies: {\em rock}, {\em paper}, and {\em scissors}. The pay-off function of players $a$ and $d$ is the constant 0. If $a$ and $d$ pick the same strategy, then neither $b$ nor $c$ is paid. If players $a$ and $d$ pick different strategies, then players $b$ and $c$ are paid according to the rules of the standard rock-paper-scissors game. In this game Nash equilibrium is only possible if $a$ and $d$ pick the same strategy. Hence, $a\rhd d$. At the same time, in any such equilibria $b$ and $c$ can have any possible combination of values. Thus, $\neg (b,c\rhd d)$. Therefore, the statement $a\rhd d \rightarrow b,c\rhd d$ is not true for this game.

\begin{wrapfigure}{l}{0.35\textwidth}
\begin{center}
\vspace{-5mm}
\scalebox{.5}{\includegraphics{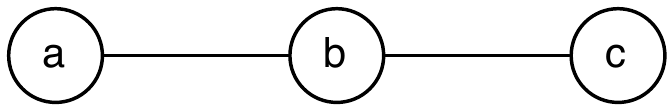}}
\vspace{0mm}
\footnotesize\caption{Dependency Graph $\Gamma_3$}\label{intro_gamma}
\vspace{-5mm}
\end{center}
\vspace{0cm}
\end{wrapfigure}

As our final example, consider the graph $\Gamma_3$ depicted in Figure~\ref{intro_gamma}. We will show that $a\rhd c\rightarrow b\rhd c$ is not true for at least one game over graph $\Gamma_3$. Indeed, consider the game in which players $a,b$, and  $c$ use real numbers as possible strategies. Players $a$ and $c$ have a constant pay-off of 0. The pay-off of the player $b$ is equal to $0$ if players $a$ and $c$ choose the same real number. Otherwise, it is equal to the number chosen by the player $b$ himself. Note that in any Nash equilibrium of this game, the strategies of players $a$ and $c$ are equal. Therefore, $a\rhd c$, but $\neg(b\rhd c)$.

The main result of this paper is a sound and complete axiomatization of all properties of functional dependence for any given dependency graph. This result is closely related to work by More and Naumov on functional dependence of secrets over hypergraphs~\cite{mn11clima}. However, the logical system presented in this paper is significantly different from theirs. A similar relation of ``rational" functional dependence without any connection to dependency graphs has been axiomatized by Naumov and Nicholls~\cite{nn12loft}. 

The counterexample that we have constructed for the game in Figure~\ref{intro_gamma} significantly relies  on the fact that  player $b$ has infinitely many strategies. However, in this paper we
 show completeness with respect to the semantics of finite games, making the result stronger.

\section{Syntax and Semantics}

The graphs that we consider in this paper contain no loops, multiple edges, or directed edges.
\begin{definition}\label{border}
For any set of vertices $U$ of a graph $(V,E)$, border ${\cal B}(U)$ is the set 
$$\{v\in U\;|\; \mbox{$(v,w)\in E$ for some $w\in V\setminus U$}\}.$$
\end{definition}
A cut $(U,W)$ of a graph $(V,E)$ is a partition $U\sqcup W$ of the set $V$. For any  vertex $v$ in a graph, by $Adj(v)$ we mean the set of all vertices adjacent to $v$. By $Adj^+(v)$ we mean the set $Adj(v)\cup\{v\}$.

\begin{definition}\label{formula}
For any graph $\Gamma=(V,E)$, by $\Phi(\Gamma)$ we mean the minimal set of formulas such that
(i) $\bot\in \Phi(\Gamma)$,
(ii) $A\rhd B\in \Phi(\Gamma)$ for each $A\subseteq V$ and $B\subseteq V$,
(iii) $\phi\rightarrow\psi\in\Phi(\Gamma)$ for each $\phi,\psi\in\Phi(\Gamma)$.
\end{definition}

\begin{definition}\label{}
By game over graph $\Gamma=(V,E)$ we mean any strategic game $G=(V,\{S_v\}_{v\in V},\{u_v\}_{v\in V})$ such that
(i) The finite set of players in the game is the set of vertices $V$,
(ii) The finite set of strategies $S_v$ of any player $v$ is an arbitrary set,
(iii) The pay-off function $u_v$ of any player $v$ only depends on the strategies of the players in $Adj^+(v)$.
\end{definition}
\noindent By $NE(G)$ we denote the set of all Nash equilibria in the game $G$. 
The next definition is the core definition of this paper. The second item in the list below gives a precise meaning of the functional dependence predicate $A\rhd B$.

\begin{definition}\label{true}
For any game $G$ over graph $\Gamma$ and any $\phi\in\Phi(\Gamma)$, we define binary relation $G\vDash \phi$ as follows
(i) $G\nvDash\bot$,
(ii) $G\vDash A\rhd B$ if ${\mathbf s}=_A{\mathbf t}$ implies ${\mathbf s}=_B {\mathbf t}$ for each ${\mathbf s},{\mathbf t}\in NE(G)$,
(iii) $G\vDash\psi_1\rightarrow\psi_2$ if $G\nvDash\psi_1$ or $G\vDash\psi_2$,
where here and everywhere below $\langle s_v\rangle_{v\in V}=_X \langle t_v\rangle_{v\in V}$ means that $s_x=t_x$ for each $x\in X$.
\end{definition}

\section{Axioms}
The following is the set of axioms of our logical system. It consists of the original Armstrong axioms and an additional Contiguity axiom that captures properties of functional dependence specific to a given graph $\Gamma$.

\begin{enumerate}
\item Reflexivity: $A\rhd B$, where $B\subseteq A$
\item Augmentation: $A\rhd B\rightarrow A,C\rhd B,C$
\item Transitivity: $A\rhd B \rightarrow (B\rhd C \rightarrow A\rhd C)$
\item Contiguity: $A,B\rhd C\rightarrow {\cal B}(U),{\cal B}(W),B\rhd C$, where $(U,W)$ is a cut of the graph such that $A\subseteq U$ and $C\subseteq W$.
\end{enumerate}
Note that the Contiguity axiom, unlike the Gateway axiom~\cite{mn11clima}, effectively requires ``double layer" divider ${\cal B}(U),{\cal B}(W)$ between sets $A$ and $C$. This is because in our setting values are assigned to the vertices and  not to the edges of the graph.

We write $\vdash_\Gamma\phi$ if $\phi\in \Phi(\Gamma)$ is provable from the combination of the axioms above and propositional tautologies in the language $\Phi(\Gamma)$ using the Modus Ponens inference rule. We write $X \vdash_\Gamma\phi$ if $\phi$ is provable using the additional set of axioms $X$. We often omit the parameter $\Gamma$ when its value is clear from the context.


\begin{lemma}\label{left mono}
$\vdash A\rhd C \rightarrow A,B \rhd C.$
\end{lemma}
\begin{proof}
Assume $A\rhd C$. By the Reflexivity axiom, $A,B\rhd A$. Thus, by the Transitivity axiom, $A,B\rhd C$. 
\end{proof}


\section{Examples}

In this section we give examples of proofs in our formal system. The soundness and the completeness of this system will be shown in the appendix.

\begin{proposition}\label{XYZ}
$\vdash_{\Gamma_1} a\rhd d\rightarrow b,c\rhd d$, where $\Gamma_1$ is the graph depicted in Figure~\ref{intro_alpha}.
\end{proposition}
\begin{proof}
Consider cut $(U,W)$ of the graph $\Gamma_1$ such that $U=\{a,b\}$ and $W=\{c,d\}$. Thus, ${\cal B}(U)=\{b\}$ and ${\cal B}(W)=\{c\}$. Therefore, by the Contiguity axiom, $a\rhd d\rightarrow b,c\rhd d$.
\end{proof}

\begin{proposition}\label{prop2}
$\vdash_{\Gamma_1} a,c\rhd d\rightarrow (d,b\rhd a\rightarrow b,c\rhd a,d)$, where $\Gamma_1$ is the graph depicted in Figure~\ref{intro_alpha}.
\end{proposition}
\begin{proof}
Assume that $a,c\rhd d$ and $d,b\rhd a$.
Consider cut $(U,W)$ of the graph $\Gamma_1$ such that $U=\{a,b\}$ and $W=\{c,d\}$. Thus, ${\cal B}(U)=\{b\}$ and ${\cal B}(W)=\{c\}$. Therefore,  by the Contiguity axiom with $A=\{a\}$, $B=\{c\}$, and $C=\{d\}$, $a,c\rhd d\rightarrow b,c\rhd d$. Thus, 
\begin{equation}\label{eq1}
 b,c\rhd d.
\end{equation}
 by the first assumption. Similarly, using the second assumption, $b,c\rhd a$. Hence, by the Augmentation axiom, 
\begin{equation}\label{eq2}
b,c\rhd a,b,c.
\end{equation}
Thus, from statement (\ref{eq1}) by the Augmentation axiom,
$a,b,c\rhd a,d$. Finally, using statement~(\ref{eq2}) and the Transitivity axiom,
$b,c\rhd a,d$.
\end{proof}

\begin{proposition}\label{prop3}
$\vdash_{\Gamma_4} a,c\rhd e\rightarrow b,c,d\rhd e$, where $\Gamma_4$ is the graph depicted in Figure~\ref{example_delta}.
\end{proposition}

\begin{wrapfigure}{l}{0.5\textwidth}
\begin{center}
\vspace{-5mm}
\scalebox{.5}{\includegraphics{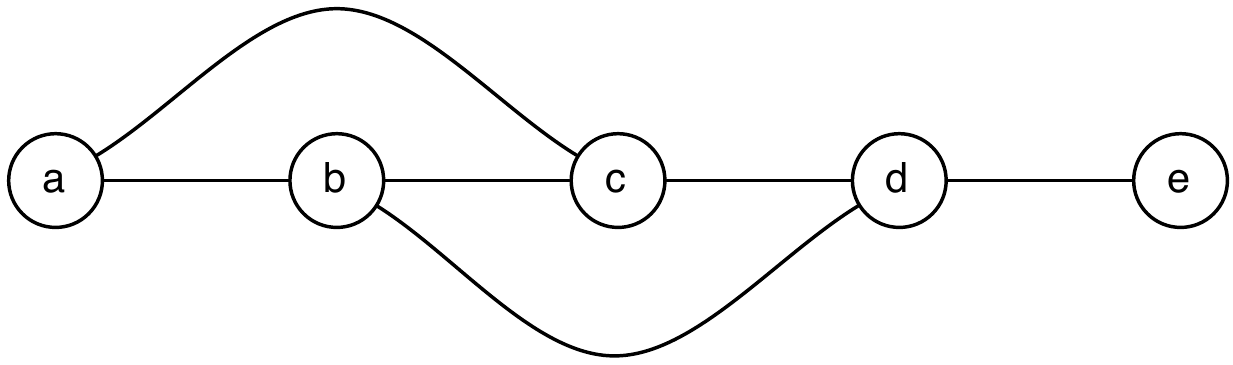}}
\vspace{-1mm}
\footnotesize\captionof{figure}{Dependency Graph $\Gamma_4$}\label{example_delta}
\vspace{-7mm}
\end{center}
\end{wrapfigure}
\noindent{\sf Proof.}
Consider cut $(U,W)$ of the graph $\Gamma_4$ such that $U=\{a,b,c\}$ and $W=\{d,e\}$. Thus, ${\cal B}(U)=\{b,c\}$ and ${\cal B}(W)=\{d\}$. Therefore, $a,c\rhd e\rightarrow b,c,d\rhd e$ by the Contiguity axiom with $A=\{a\}$, $B=\{c\}$, and $C=\{e\}$.
\qed

\pagebreak

\begin{proposition}\label{prop4}
$\vdash_{\Gamma_5} a\rhd b\rightarrow (b\rhd c\rightarrow (c\rhd a \rightarrow d,e,f\rhd a,b,c))$, where
$\Gamma_5$ is depicted in Figure~\ref{example_epsilon}.
 
\end{proposition}

\noindent{\sf Proof.}
Assume $a\rhd b$, $b\rhd c$, and $c\rhd a$.
Consider cut $(U,W)$ of the graph $\Gamma_5$ such that $U=\{c,f\}$ and $W=\{a,b,d,e\}$. Thus, ${\cal B}(U)=\{f\}$ and ${\cal B}(W)=\{d,e\}$. Therefore, by the Contiguity axiom with $A=\{c\}$, $B=\varnothing$, and $C=\{a\}$, $c\rhd a\rightarrow d,e,f\rhd a$. Hence, $d,e,f\rhd a$ by the third assumption. Similarly, one can show $d,e,f\rhd b$, and $d,e,f\rhd c$. By applying the Augmentation axiom to the last three statements,
$
d,e,f\rhd a,d,e,f,
$
and
$
a,d,e,f\rhd a,b,d,e,f,
$
and
$
a,b,d,e,f\rhd a,b,c.
$
Therefore, $d,e,f\rhd a,b,c$ by the Transitivity axiom applied twice.
\qed

\begin{wrapfigure}{r}{0.35\textwidth}
\begin{center}
\vspace{-2mm}
\scalebox{.5}{\includegraphics{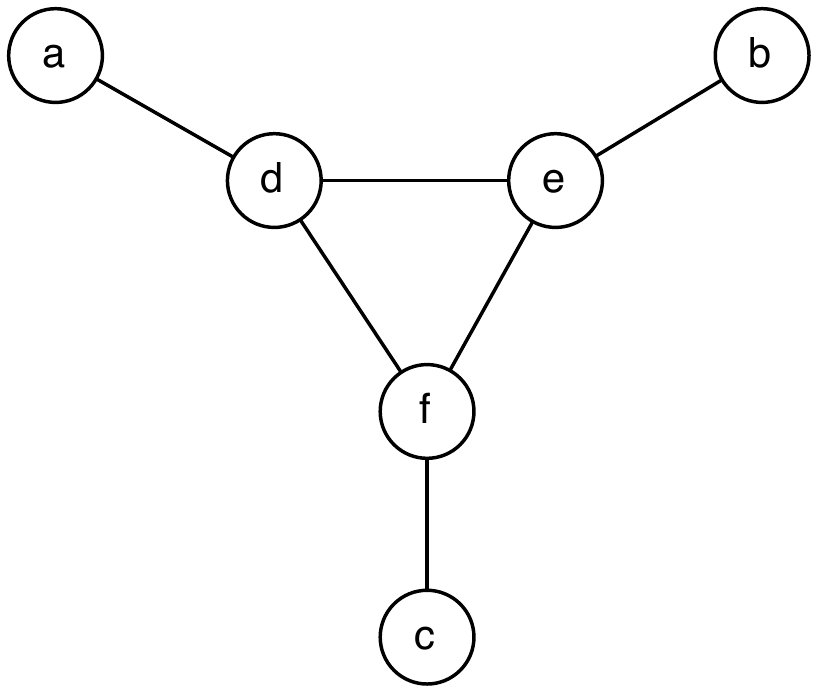}}
\vspace{0mm}
\footnotesize\caption{Dependency Graph $\Gamma_5$}\label{example_epsilon}
\vspace{-7mm}
\end{center}
\vspace{0cm}
\end{wrapfigure}

Proposition~\ref{prop2} and Proposition~\ref{prop4} are special cases of a more general principle. We will say that a subset of vertices is {\em sparse} if the shortest path between any two vertices in this subset contains at least three edges. The general principle states that if $W$ is a sparse subset of vertices in the graph $(V,E)$ and each vertex $w\in W$ is functionally determined by the set $V\setminus\{w\}$, then the subset $V\setminus W$ functionally determines the subset $W$:    
$$
\bigwedge_{w\in W}\left((V\setminus \{w\})\right)\rhd w \rightarrow (V\setminus W)\rhd W.
$$
For example, the set $\{a,d\}$ in the graph $\Gamma_1$ depicted in Figure~\ref{intro_alpha} is sparse. Due to the general principle, 
$
a,b,c\rhd d \rightarrow(d,c,b\rhd a \rightarrow b,c\rhd a,d).
$
Thus, by Lemma~\ref{left mono},
$
a,c\rhd d \rightarrow(d,b\rhd a \rightarrow b,c\rhd a,d),
$
which is the statement of Proposition~\ref{prop2}. In the case of Proposition~\ref{prop4}, the sparse set is $\{a,b,c\}$. The proof of the general principle is similar to the proof of Proposition~\ref{prop4}.

\section{Soundness}

In this section, we prove soundness of our logical system by proving soundness of each of our four axioms. The proof of completeness can be found in~\cite{hn13arxiv}.

\begin{lemma}[reflexivity]\label{} 
$G\vDash A\rhd B$ for each game $G$ over a graph $\Gamma=(V,E)$
and each $B\subseteq A\subseteq V$.
\end{lemma}
\begin{proof}
For any ${\mathbf s},{\mathbf t}\in NE(G)$, if ${\mathbf s}=_A {\mathbf t}$, then ${\mathbf s}=_B {\mathbf t}$ because $A\subseteq B$.
\end{proof}

\begin{lemma}[augmentation]\label{} 
If $G\vDash A\rhd B$, then $G\vDash A,C\rhd B,C$ for each game $G$ over a graph $\Gamma=(V,E)$
and each $A,B,C\subseteq V$.
\end{lemma}
\begin{proof}
Suppose that $G\vDash A\rhd B$ and consider any ${\mathbf s},{\mathbf t}\in NE(G)$ such that ${\mathbf s}=_{A,C}{\mathbf t}$. We will show that ${\mathbf s}=_{B,C}{\mathbf t}$. Indeed, ${\mathbf s}=_{A,C}{\mathbf t}$ implies that ${\mathbf s}=_{A}{\mathbf t}$ and ${\mathbf s}=_{C}{\mathbf t}$. Thus, ${\mathbf s}=_{B}{\mathbf t}$ by the assumption $G\vDash A\rhd B$. Therefore, ${\mathbf s}=_{B,C}{\mathbf t}$.
\end{proof}

\begin{lemma}[transitivity]\label{}
If $G\vDash A\rhd B$ and $G\vDash B\rhd C$, then $G\vDash A\rhd C$ for each game $G$ over a graph $\Gamma=(V,E)$
and each $A,B,C\subseteq V$.
\end{lemma}
\begin{proof}
Suppose that $G\vDash A\rhd B$ and $G\vDash B\rhd C$. Consider any ${\mathbf s},{\mathbf t}\in NE(G)$ such that ${\mathbf s}=_{A}{\mathbf t}$. We will show that ${\mathbf s}=_{C}{\mathbf t}$. Indeed, ${\mathbf s}=_{B}{\mathbf t}$ due to the first assumption. Hence, by the second assumption, ${\mathbf s}=_{C}{\mathbf t}$.
\end{proof}

\begin{lemma}[contiguity]\label{}
If $G\vDash A,B\rhd C$, then $G\vDash {\cal B}(S),{\cal B}(T),B\rhd C$,
for each game $G=(V,E)$ over a graph $\Gamma$, each cut $(U,W)$ of $\Gamma$, and each $A\subseteq U$, $B\subseteq V$, and $C\subseteq W$.
\end{lemma}
\begin{proof}
Suppose that $G\vDash A,B\rhd C$. Consider any ${\mathbf s}=\langle s_v\rangle_{v\in V}\in NE(G)$ and ${\mathbf t}=\langle t_v\rangle_{v\in V}\in NE(G)$ such that ${\mathbf s}=_{{\cal B}(U),{\cal B}(W),B}{\mathbf t}$. We will prove that ${\mathbf s}=_C {\mathbf t}$. Indeed, consider strategy profile ${\mathbf e}=\langle e_v\rangle_{v\in V}$ such that 
$$
e_v=
\left\{
\begin{array}{ll}
s_v   & \mbox{if $v\in U$,}  \\
t_v  & \mbox{if $v\in W$}.  
\end{array}
\right.
$$ 
We will first prove that ${\mathbf e}\in NE(G)$. Assuming the opposite, let $v\in V$ be a player in the game $G$ that can increase his pay-off by changing strategy in profile ${\mathbf e}$. Without loss of generality, let $v\in U$. Then, ${\mathbf e}=_{Adj(v)\cup{\{v\}}}{\mathbf s}$. Thus, player $v$  can also increase his pay-off by changing strategy in profile ${\mathbf s}$, which is a contradiction with the choice of ${\mathbf s}\in NE(G)$.

Note that ${\mathbf e}=_{U,B}{\mathbf s}$ and ${\mathbf e}=_{W,B}{\mathbf t}$. Thus, ${\mathbf e}=_{A,B}{\mathbf s}$ and ${\mathbf e}=_{C}{\mathbf s}$. Hence, ${\mathbf e}=_{C}{\mathbf s}$ by the assumption $G\vDash A,B\rhd C$. Therefore, ${\mathbf s}=_{C}{\mathbf e} =_{C}{\mathbf t}$.
\end{proof}

\section{Conclusion}

In this paper, we have described a sound and complete logical system for functional dependence in strategic games over a fixed dependency graph. The dependency graph puts restrictions on the type of pay-off functions that can be used in the game. If no such restrictions are imposed, then the logical system for functional dependence in strategic games is just the set of original Armstrong axioms. This statement follows from our results since the absence of restrictions corresponds to the case of a complete (in the graph theory sense) dependency graph. In the case of a complete graph, the Contiguity axiom follows from the Armstrong axioms because for any cut $(U,W)$, the set ${\cal B}(U)\cup{\cal B}(W)$ is the set of all vertices in the graph.

\bibliography{../sp}

\begin{thebibliography}{10}
\providecommand{\bibitemdeclare}[2]{}
\providecommand{\surnamestart}{}
\providecommand{\surnameend}{}
\providecommand{\urlprefix}{Available at }
\providecommand{\url}[1]{\texttt{#1}}
\providecommand{\href}[2]{\texttt{#2}}
\providecommand{\urlalt}[2]{\href{#1}{#2}}
\providecommand{\doi}[1]{doi:\urlalt{http://dx.doi.org/#1}{#1}}
\providecommand{\bibinfo}[2]{#2}

\bibitemdeclare{incollection}{a74}
\bibitem{a74}
\bibinfo{author}{W.~W. \surnamestart Armstrong\surnameend}
  (\bibinfo{year}{1974}): \emph{\bibinfo{title}{Dependency structures of data
  base relationships}}.
\newblock In: {\sl \bibinfo{booktitle}{Information processing 74 ({P}roc.
  {IFIP} {C}ongress, {S}tockholm, 1974)}}, \bibinfo{publisher}{North-Holland},
  \bibinfo{address}{Amsterdam}, pp. \bibinfo{pages}{580--583}.

\bibitemdeclare{inproceedings}{bfh77}
\bibitem{bfh77}
\bibinfo{author}{Catriel \surnamestart Beeri\surnameend},
  \bibinfo{author}{Ronald \surnamestart Fagin\surnameend} \&
  \bibinfo{author}{John~H. \surnamestart Howard\surnameend}
  (\bibinfo{year}{1977}): \emph{\bibinfo{title}{A complete axiomatization for
  functional and multivalued dependencies in database relations}}.
\newblock In: {\sl \bibinfo{booktitle}{SIGMOD '77: Proceedings of the 1977 ACM
  SIGMOD international conference on Management of data}},
  \bibinfo{publisher}{ACM}, \bibinfo{address}{New York, NY, USA}, pp.
  \bibinfo{pages}{47--61}, \doi{10.1145/509404.509414}.

\bibitemdeclare{article}{egg06eccc}
\bibitem{egg06eccc}
\bibinfo{author}{Edith \surnamestart Elkind\surnameend},
  \bibinfo{author}{Leslie~Ann \surnamestart Goldberg\surnameend} \&
  \bibinfo{author}{Paul~W. \surnamestart Goldberg\surnameend}
  (\bibinfo{year}{2006}): \emph{\bibinfo{title}{Nash Equilibria in Graphical
  Games on Trees Revisited}}.
\newblock {\sl \bibinfo{journal}{Electronic Colloquium on Computational
  Complexity (ECCC)}} (\bibinfo{number}{005}).

\bibitemdeclare{inproceedings}{egg07ec}
\bibitem{egg07ec}
\bibinfo{author}{Edith \surnamestart Elkind\surnameend},
  \bibinfo{author}{Leslie~Ann \surnamestart Goldberg\surnameend} \&
  \bibinfo{author}{Paul~W. \surnamestart Goldberg\surnameend}
  (\bibinfo{year}{2007}): \emph{\bibinfo{title}{Computing good {N}ash
  equilibria in graphical games}}.
\newblock In \bibinfo{editor}{Jeffrey~K. \surnamestart
  MacKie-Mason\surnameend}, \bibinfo{editor}{David~C. \surnamestart
  Parkes\surnameend} \& \bibinfo{editor}{Paul \surnamestart
  Resnick\surnameend}, editors: {\sl \bibinfo{booktitle}{ACM Conference on
  Electronic Commerce}}, \bibinfo{publisher}{ACM}, pp.
  \bibinfo{pages}{162--171}, \doi{10.1145/1250910.1250935}.

\bibitemdeclare{book}{guw09}
\bibitem{guw09}
\bibinfo{author}{Hector \surnamestart Garcia-Molina\surnameend},
  \bibinfo{author}{Jeffrey \surnamestart Ullman\surnameend} \&
  \bibinfo{author}{Jennifer \surnamestart Widom\surnameend}
  (\bibinfo{year}{2009}): \emph{\bibinfo{title}{Database Systems: The Complete
  Book}}, \bibinfo{edition}{second} edition.
\newblock \bibinfo{publisher}{Prentice-Hall}.

\bibitemdeclare{article}{hn13arxiv}
\bibitem{hn13arxiv}
\bibinfo{author}{Kristine \surnamestart Harjes\surnameend} \&
  \bibinfo{author}{Pavel \surnamestart Naumov\surnameend}
  (\bibinfo{year}{2013}): \emph{\bibinfo{title}{Functional Dependence in
  Strategic Games}}.
\newblock {\sl \bibinfo{journal}{CoRR}} \bibinfo{volume}{arXiv:1302.0447
  [math.LO]}.

\bibitemdeclare{inproceedings}{kls01uai}
\bibitem{kls01uai}
\bibinfo{author}{Michael~J. \surnamestart Kearns\surnameend},
  \bibinfo{author}{Michael~L. \surnamestart Littman\surnameend} \&
  \bibinfo{author}{Satinder~P. \surnamestart Singh\surnameend}
  (\bibinfo{year}{2001}): \emph{\bibinfo{title}{Graphical Models for Game
  Theory}}.
\newblock In \bibinfo{editor}{Jack~S. \surnamestart Breese\surnameend} \&
  \bibinfo{editor}{Daphne \surnamestart Koller\surnameend}, editors: {\sl
  \bibinfo{booktitle}{UAI}}, \bibinfo{publisher}{Morgan Kaufmann}, pp.
  \bibinfo{pages}{253--260}.

\bibitemdeclare{inproceedings}{lks01nips}
\bibitem{lks01nips}
\bibinfo{author}{Michael~L. \surnamestart Littman\surnameend},
  \bibinfo{author}{Michael~J. \surnamestart Kearns\surnameend} \&
  \bibinfo{author}{Satinder~P. \surnamestart Singh\surnameend}
  (\bibinfo{year}{2001}): \emph{\bibinfo{title}{An Efficient, Exact Algorithm
  for Solving Tree-Structured Graphical Games}}.
\newblock In \bibinfo{editor}{Thomas~G. \surnamestart Dietterich\surnameend},
  \bibinfo{editor}{Suzanna \surnamestart Becker\surnameend} \&
  \bibinfo{editor}{Zoubin \surnamestart Ghahramani\surnameend}, editors: {\sl
  \bibinfo{booktitle}{NIPS}}, \bibinfo{publisher}{MIT Press}, pp.
  \bibinfo{pages}{817--823}.

\bibitemdeclare{inproceedings}{mn11clima}
\bibitem{mn11clima}
\bibinfo{author}{Sara~Miner \surnamestart More\surnameend} \&
  \bibinfo{author}{Pavel \surnamestart Naumov\surnameend}
  (\bibinfo{year}{2011}): \emph{\bibinfo{title}{The Functional Dependence
  Relation on Hypergraphs of Secrets}}.
\newblock In \bibinfo{editor}{Jo{\~a}o \surnamestart Leite\surnameend},
  \bibinfo{editor}{Paolo \surnamestart Torroni\surnameend},
  \bibinfo{editor}{Thomas \surnamestart {\AA}gotnes\surnameend},
  \bibinfo{editor}{Guido \surnamestart Boella\surnameend} \&
  \bibinfo{editor}{Leon \surnamestart van~der Torre\surnameend}, editors: {\sl
  \bibinfo{booktitle}{CLIMA}}, {\sl \bibinfo{series}{Lecture Notes in Computer
  Science}} \bibinfo{volume}{6814}, \bibinfo{publisher}{Springer}, pp.
  \bibinfo{pages}{29--40}, \doi{10.1007/978-3-642-22359-4\_3}.

\bibitemdeclare{inproceedings}{nn12loft}
\bibitem{nn12loft}
\bibinfo{author}{Pavel \surnamestart Naumov\surnameend} \&
  \bibinfo{author}{Brittany \surnamestart Nicholls\surnameend}
  (\bibinfo{year}{2012}): \emph{\bibinfo{title}{Rationally Functional
  Dependence}}.
\newblock In: {\sl \bibinfo{booktitle}{10th Conference on Logic and the
  Foundations of Game and Decision Theory (LOFT)}}.

\end{thebibliography}

\end{document}